\documentclass[11pt,a4paper]{article}
\usepackage{fullpage}
\usepackage{amsmath}
\usepackage{amssymb}
\usepackage{wasysym}
\usepackage{stmaryrd}
\usepackage{graphicx}
\usepackage{psfrag}
\usepackage{color}

\newtheorem{definition}{Definition}

\newtheorem{lemma}{Lemma}
\newenvironment{proof}{\noindent \begin{rm}{\textbf{Proof.} }}{\hspace*{\fill}$\Box$\par\end{rm}}

\definecolor{gris}{gray}{0.50}

\newcommand{\id}{\mbox{\sf Id}}
\newcommand{\MF}{\mbox{\tt ME}}

\newcommand{\parent}{\mbox{\it p}}
\newcommand{\lab}{\mbox{\rm $\ell$}}
\newcommand{\TS}{\mbox{\it size}}
\newcommand{\m}{\mbox{\it mwe}}

\newcommand{\last}{\mbox{\sf last}}
\newcommand{\Child}{\mbox{\rm Child}}
\newcommand{\SizeC}{\mbox{\rm SizeC}}

\newcommand{\Cycle}{\mbox{\rm Cycle}}
\newcommand{\LabelR}{\mbox{\rm Label$_R$}}
\newcommand{\LabelNd}{\mbox{\rm Label$_{Nd}$}}
\newcommand{\Label}{\mbox{\rm Label}}
\newcommand{\Isleaf}{\mbox{\rm Leaf}}
\newcommand{\Heavy}{\mbox{\rm Heavy}}
\newcommand{\Light}{\mbox{\rm Light}}

\newcommand{\NeedReorientation}{\mbox{\rm NeedReorientation}}
\newcommand{\EndReorientation}{\mbox{\rm EndReorientation}}
\newcommand{\TreeChange}{\mbox{\rm TreeMerg}}
\newcommand{\NewFragment}{\mbox{\rm NewFrag}}
\newcommand{\IsMinEnabled}{\mbox{\rm MinEnabled}}
\newcommand{\Enabled}{\mbox{\rm Enabled}}
\newcommand{\MinEdge}{\mbox{\rm MinEdge}}
\newcommand{\FarLca}{\mbox{\rm FarLca}}
\newcommand{\FarLcaC}{\mbox{\rm FarLcaChild}}

\newcommand{\path}{\mbox{\tt path}}

\newcommand{\Lca}{\mbox{\rm $nca$}}
\newcommand{\FC}{\mbox{\rm MergeChild}}
\newcommand{\FA}{\mbox{\rm MergeAdj}}
\newcommand{\Fusion}{\mbox{\rm MergeEdge}}
\newcommand{\UC}{\mbox{\rm RecoverChild}}
\newcommand{\UA}{\mbox{\rm RecoverAdj}}
\newcommand{\Update}{\mbox{\rm RecoverEdge}}
\newcommand{\RRoot}{\mbox{\rm R$_{\odot}$}} 	%R$_{Root}$}}
\newcommand{\RLC}{\mbox{\rm R$_{\ell}$}} %Label}$}}
\newcommand{\RMin}{\mbox{\rm R$_{Min}$}}
\newcommand{\RC}{\mbox{\rm R$_{\lightning}$}} %Corr}$}}

\newcommand{\RF}{\mbox{\rm R$_{\vartriangleright \vartriangleleft}$}}
\newcommand{\REnd}{\mbox{\rm R$_{\hexagon}$}}

\newcommand{\MST}{\mbox{\tt MST}}
\newcommand{\LabA}{\mbox{\tt NCA-L}}

\begin{document}
\title{Fast Self-Stabilizing Minimum Spanning Tree Construction\\
 \small{Using Compact Nearest Common Ancestor Labeling Scheme}}

\author{
L{\'e}lia Blin$^{1,3}$
\and
Shlomi Dolev$^4$
\and
Maria Gradinariu Potop-Butucaru$^{2,3,5}$
\and
Stephane Rovedakis$^{1,6}$
}

% \institute{
% Universit\'e d'Evry-Val d'Essonne, 91000 Evry, France.\\
% \and
% Universit\'e Pierre \& Marie Curie - Paris 6, 75005 Paris, France.\\
% \and
% LIP6-CNRS UMR 7606, France.\\
% \email{\{lelia.blin,maria.gradinariu\}@lip6.fr}
% \and
% Department of Computer Science, Ben-Gurion University of the Negev,\\
% Beer-Sheva, 84105, Israel.\\
% \email{dolev@cs.bgu.ac.il}
% \and
% INRIA REGAL, France.\\
% \and
% Laboratoire IBISC-EA 4526, 91000 Evry, France.\\
% \email{stephane.rovedakis@ibisc.fr}
% }

\footnotetext[1]{Universit\'e d'Evry-Val d'Essonne, 91000 Evry, France.}
\footnotetext[2]{Universit\'e Pierre \& Marie Curie - Paris 6, 75005 Paris, France.}
\footnotetext[3]{LIP6-CNRS UMR 7606, France. \texttt{\{lelia.blin,maria.gradinariu\}@lip6.fr}}
\footnotetext[4]{Department of Computer Science, Ben-Gurion University of the Negev,\\ Beer-Sheva, 84105, Israel. \texttt{dolev@cs.bgu.ac.il}}
\footnotetext[5]{INRIA REGAL, France.}
\footnotetext[6]{Laboratoire IBISC-EA 4526, 91000 Evry, France. \texttt{stephane.rovedakis@ibisc.fr}}

\maketitle

%======================================================================
% ABSTRACT
%======================================================================
\begin{abstract}
We present a novel self-stabilizing algorithm for minimum spanning tree (MST) construction. 
The space complexity of our solution is $O(\log^2n)$ bits and it converges in $O(n^2)$ rounds. 
Thus, this algorithm improves the convergence time of all previously known self-stabilizing asynchronous MST algorithms by 
a multiplicative factor $\Theta(n)$, to the price of increasing the best known space complexity by a factor $O(\log n)$. 
The main ingredient used in our algorithm is the design, for the first time in self-stabilizing settings, of a labeling scheme  
for computing the nearest common ancestor with only $O(\log^2n)$ bits. 
%Although this label size is not the best possible one (there exist nearest common ancestor labeling schemes using 
%labels of size $O(\log n)$ bits), our scheme is, up to our knowledge, the first known informative labeling scheme constructed in a self-stabilizing manner. 
\end{abstract}

% \thispagestyle{empty}
% \setcounter{page}{0}
% \newpage 

%======================================================================
% INTRO
%======================================================================
\section{Introduction}
\label{sec:intro}
Since its introduction in a centralized context~\cite{Prim57,Kruskal56}, the minimum spanning tree (or MST) problem gained a 
benchmark status in distributed computing thanks to the seminal work of Gallager, Humblet and Spira~\cite{GallagerHS83}.  

The emergence of large scale and dynamic systems, often subject to transient faults, revives the study of scalable and self-stabilizing algorithms.
A \emph{scalable} algorithm does not rely on any global parameter of the system (\emph{e.g.} upper bound on the  number of nodes or the diameter). 
\emph{Self-stabilization} introduced first by Dijkstra in~\cite{D74j} and later publicized by several books~\cite{Dolev00,Tel94} deals 
with the ability of a system to recover from catastrophic situation (i.e., the global state may be arbitrarily far from a legal state)
without external (\emph{e.g.} human) intervention in finite time.

%This algorithm is fast and does not rely on any global knowledge.
%Therefore it is appealing for large scale networks.
%According to this criteria the Gallager, 
%Humblet and Spira algorithm steams from its scalability.  
Although there already exists self-stabilizing solutions for the MST construction, none of them considered the extension of   
the Gallager, Humblet and Spira algorithm (GHS) to self-stabilizing settings.  
Interestingly, this algorithm unifies the best properties for designing large scale 
MSTs: it is fast and totally decentralized and it does not 
rely on any global parameter of the system. Our work proposes an extension of this algorithm to self-stabilizing settings. Our extension uses only logarithmic memory and  
preserves all the good characteristics of the original solution in terms of 
convergence time and scalability. 

Gupta and Srimani presented in~\cite{AntonoiuS97} the first self-stabilizing algorithm for the MST problem. 
%It applies on graphs with uniquely identified nodes, where edges have integer edge weights, and a weight can appear at most once in the whole network. 
The MST construction is based on the computation of all 
shortest paths (for a certain cost function) 
between all pairs of nodes. While executing the algorithm, every node stores the cost of 
all paths from it to all the other nodes. To implement this algorithm, the authors assume that every node knows the number $n$ of nodes in 
the network, and that the identifiers of the nodes are in $\{1,\dots,n\}$. Every node $u$ stores the weight of the edge $e_{u,v}$ placed in the 
MST for each node $v\neq u$. Therefore the algorithm requires $\Omega(\sum_{v\neq u}\log w(e_{u,v}))$ bits of memory at node $u$. Since all 
the weights are distinct integers, the memory requirement at each node is $\Omega(n\log n)$ bits.
The main drawback of this solution is its lack of scalability since each node has to know and maintain information for all the nodes in the system. 
Note also that the time complexity announced by the authors, $O(n)$
stays only in the particular synchronous settings considered by the
authors. In asynchronous setting the complexity is $\Omega(n^2)$ rounds.
A different approach for the message-passing model, was proposed by 
Higham and Liang~\cite{HighamL01}. 
%Their work applies to undirected connected graphs with unique integer edge weights and unique node identifiers. In the mode 
%knows only an upper bound on the number of nodes in the system. 
The algorithm performs roughly as follows: every edge checks 
whether it eventually belongs to the MST or not. 
To this end, every non tree-edge $e$ floods the network to find 
a potential cycle, and when $e$ receives its own message back along a cycle, it uses the information collected by this message (\emph{i.e.}, 
the maximum edge weight of the traversed cycle) to decide whether $e$ could potentially be in the MST or not. If the edge $e$ has not 
received its message back after the time-out interval, it decides to become tree edge. The memory used by each node is $O(\log n)$ bits, but the information exchanged between neighboring nodes is of size $O(n \log n)$ bits, 
thus only slightly improving that of \cite{AntonoiuS97}. This solution
also assume that each node has 
access to a global parameter of the system: the diameter. Its
computation 
is expensive in large scale systems and becomes even harder in dynamic settings.
The time complexity of this approach is $O(mD)$ rounds where $m$ and $D$ are the number of edges and the diameter of the network respectively, i.e., $O(n^3)$ rounds in the worst case.
% Note also that the time complexity announced by the authors, $O(n)$
% stays only in the particular synchronous settings considered by the
% authors. In asynchronous setting the complexity is $O(n^2)$ rounds.

In \cite{BPRT09c} we proposed a  self-stabilizing loop-free algorithm for the MST problem. 
Contrary to previous self-stabilizing MST protocols, this algorithm does not make
any assumption on the network size (including upper bounds) or the unicity of the
edge weights. The proposed solution improves on the memory space usage since each
participant needs only $O(\log n)$ bits while preserving the same time complexity as the algorithm in \cite{HighamL01}.  

Clearly, in the self-stabilizing implementation of the MST algorithms there is a trade-off between 
the memory complexity and their time complexity (see Table \ref{tableresume}, where a boldface denotes the most useful (or efficient) feature for a particular criterium). The challenge we address in this paper is to design fast 
and scalable self-stabilizing MST with little memory. Our approach brings together two worlds: the time efficient 
MST constructions and the memory compact informative labeling schemes. 
Therefore, we extend
the GHS algorithm to self-stabilizing settings and keep compact its 
memory space by using a self-stabilizing extension 
of the nearest common ancestor labeling scheme of \cite{AGKR02}.
Note that labeling schemes have already been used in 
order to infer a broad set of information such as
vertex adjacency, distance, tree ancestry or tree routing~\cite{BeinDV05},  
however none of these schemes have been studied in self-stabilizing
settings (except the last one).\\

%\subsection{Minimum spanning tree problem}

%\paragraph{Related works}
% \textcolor{red}{synchrone, discute modele et complexite}\\

%\textcolor{red}{rajout de paragraphe sur l'article de disc l'an dernier}

Our contribution is therefore twofold. 
We propose for the first time in self-stabilizing settings a $O(log^2n)$ bits 
scheme for computing the nearest common ancestor. 
Furthermore, based on this scheme, we describe a new 
% silent 
self-stabilizing algorithm for the MST
problem. Our algorithm does not make any assumption on the network size
(including upper bounds) or the existence of an a priori known root. Moreover,
our solution is the best space/time compromise over the existing self-stabilizing MST solutions.
The convergence time is $O(n^2)$ asynchronous rounds 
and the memory space per node is 
$O(\log^2 n)$ bits. Interestingly, our work is the first to prove the effectiveness 
of an informative labeling scheme in self-stabilizing settings and therefore 
opens a wide research path in this direction.   
 
\begin{table}[t]
\begin{center}
\scalebox{1}
{
\begin{tabular}{|l|c|c|c|}
\hline
 & a priori knowledge & space complexity & convergence time\\
\hline
 \cite{AntonoiuS97} & network size and & $O(n \log n)$ & $\Omega(n^2)$ \\
 &  the nodes in the network & & \\
\hline
 \cite{HighamL01}& upper bound  on diameter &$O(\log n)$& $O(n^3)$\\
 & & +messages of size $O(n\log n)$ & \\
\hline
\cite{BPRT09c}  & \textbf{none} &  $\mathbf{O(\log n)}$ & $O(n^3)$\\
\hline
This paper & \textbf{none} & $O(\log^2 n)$ & $\mathbf{O(n^2)}$\\
\hline
\end{tabular}
}
\caption{\small Distributed Self-Stabilizing algorithms for the MST problem}
\label{tableresume}
\end{center}
\end{table}

%======================================================================
\section{Model and notations} 
\label{sec:model}
%======================================================================

% \paragraph{System model}

We consider an undirected weighted connected network $G=(V,E,w)$ where $V$ is the set of nodes, $E$ is the set of edges and $w: E \rightarrow {\mathbb R^+}$ is a positive cost function. 
Nodes represent processors and edges represent bidirectional communication links. Additionally, we consider that $G=(V,E,w)$ is a network in which the weight of the communication links may change value.

The processors asynchronously execute their programs consisting of a set of variables and a finite set of rules. The variables are part of the shared register which is used to communicate with the neighbors. A processor can read and write its own registers and can read the shared registers of its neighbors. 
%So, variables of a processor can be accessed by the processor and its neighbors. 
Each processor executes a program consisting of a sequence of guarded rules. Each \emph{rule} contains a \emph{guard} (Boolean expression over the variables of a node and its neighborhood) and an \emph{action} (update of the node variables only). Any rule whose guard is \emph{true} is said to be \emph{enabled}. A node with one or more enabled rules is said to be \emph{privileged} and may make a \emph{move} executing the action corresponding to the chosen enabled rule.

A {\it local state} of a node is the value of the local variables of the node and the state of its program counter. A {\it configuration} of the system $G=(V,E)$ is the cross product of the local states of all nodes in the system. The transition from a configuration to the next one is produced by the execution of an action at a node. A {\it computation} of the system is defined as a \emph{weakly fair, maximal} sequence of configurations, $e=(c_0, c_1, \ldots c_i, \ldots)$, where each configuration $c_{i+1}$ follows from $c_i$ by the execution of a single action of at least one node. During an execution step, one or more processors execute an action and a processor may take at most one action. \emph{Weak fairness} of the sequence means that if any action in $G$ is continuously enabled along the sequence, it is eventually chosen for execution. \emph{Maximality} means that the sequence is either infinite, or it is finite and no action of $G$ is enabled in the final global state.

%%%%%%%
%\paragraph{Faults and self-stabilization}
In the sequel we consider the system can start in any configuration. That is, the local state of a node can be corrupted. Note that we don't make any assumption on the bound of corrupted nodes. In the worst case all the nodes in the system may start in a corrupted configuration. In order to tackle these faults we use self-stabilization techniques.

\begin{definition}[self-stabilization]
Let $\mathcal{L_{A}}$ be a non-empty \emph{legitimacy predicate}\footnote{A legitimacy predicate is defined over the configurations of a system and is an indicator of its correct behavior.} of an algorithm $\mathcal{A}$ with respect to a specification predicate $Spec$ such that every configuration satisfying $\mathcal{L_{A}}$ satisfies $Spec$. Algorithm $\mathcal{A}$ is \emph{self-stabilizing} with respect to $Spec$ iff the following two conditions hold:\\
\textsf{(i)} Every computation of $\mathcal{A}$ starting from a configuration satisfying $\mathcal{L_A}$ preserves $\mathcal{L_A}$ (\emph{closure}).  \\
\textsf{(ii)} Every computation of $\mathcal{A}$ starting from an arbitrary configuration contains a configuration that satisfies $\mathcal{L_A}$ (\emph{convergence}).
\end{definition}

\section{Overview of our solution}
We propose to extend the Gallager, Humblet and Spira (GHS) algorithm, 
\cite{GallagerHS83}, 
to self-stabilizing settings via a compact informative labeling scheme.
Thus, the resulting solution presents several advantages 
appealing for large scale systems: it is compact since it 
uses only logarithmic memory in the size of the network, 
it scales well since it does not rely on any global 
parameter of the system, it is fast --- its time complexity is the 
better known in self-stabilizing settings. Additionally, it self-recovers 
from any transient fault.

The central notion in the GHS approach is the notion of \textit{fragment}.
A fragment is a partial spanning tree of the graph, i.e., a fragment is a tree which spans a subset of nodes.
Note that a fragment can be limited to a single node.
An outgoing edge of a fragment $F$ is an edge with a unique endpoint in $F$.
The minimum-weight outgoing edge of a fragment $F$ is denoted in the following as \MF$_F$.
In the GHS construction, initially each node is a fragment.
For each fragment $F$, the GHS algorithm in~\cite{GallagerHS83} identifies 
the \MF$_F$ and merges the two fragments endpoints of \MF$_F$. 
Note that, with this scheme, more than two fragments may be merged concurrently. The merging process is recursively repeated until a single 
fragment remains. The result is a MST. The above approach is often 
called ``blue rule" for MST construction.

This approach is particularly appealing 
when transient faults yield to a forest of fragments (which are sub-trees of a MST).
The direct application of the blue rule allows the system to reconstruct a MST and to recover from faults which have divided the existing MST. 
However, when more severe faults hit the system 
the process variables may be corrupted leading to 
a configuration of the network where the 
set of fragments are not sub-trees of some MST. 
That is, it may be a spanning tree but not of minimum weight, 
or it can contain cycles. In this case, 
the application of the blue rule only is not sufficient to reconstruct a MST. 
To overcome this difficulty, we combine the blue rule 
with another method, referred in the literature as the ``red rule".
The red rule removes the heaviest edge from every cycle. 
The resulting configuration contains a MST. 
We use the red rule as follows: given a 
spanning tree $T$ of $G$, every edge $e$ of $G$ that is 
not in $T$ is added to $T$, thus creating a (unique) cycle in $T\cup\{e\}$.
This cycle is called a \textit{fundamental cycle}, denoted by $C_e$. 
%(Formally, this cycle depends on $T$; Nevertheless no confusion should arise from omitting $T$ in the notation $C_e$). 
If $e$ is not the edge of maximum weight in $C_e$, then, according to the red rule, there exists an edge $f \neq e$ in $C_e$ with $w(f)>w(e)$. 
The edge of maximum weight can be removed since it is not part of any MST. 

Our \MST\/ construction combines both the blue and red rules.
The blue rule application needs that each node identifies its own fragment.
The red rule requires that nodes identify the fundamental cycle corresponding to every adjacent non-tree-edge.
In both cases, we use a self-stabilizing 
labeling scheme, called $\LabA$, which 
provides at each node a distinct informative label such that  
the nearest common ancestor 
of two nodes can be identified based only on the 
labels of these nodes (see Section~\ref{sec:label}). 
Thus, the advantage of this labeling 
is twofold. First the labeling helps nodes to identify their fragments.
Second, given any non-tree edge $e=(u,v)$, the path in 
the tree going from $u$ to the nearest common ancestor of $u$ and $v$, then 
from there to $v$, and finally back to $u$ by traversing 
$e$, constitute the fundamental cycle $C_e$.

To summarize, our algorithm will use the blue rule to construct a spanning tree, and the red rule to recover from invalid configurations. In both cases, it 
uses our algorithm \LabA\/ 
%the nearest-common ancestor labeling scheme 
to identify both fragments and fundamental cycles. Note that, in~\cite{ParkMHT90,PMHT92} distributed algorithms using the blue and red rules to construct a MST in a dynamic network are proposed, however these algorithms are not self-stabilizing.
%Of course, the labeling scheme it self can be corrupted. Our algorithm handles this as follows..................

\paragraph{Variables used by \LabA\ and \MST\/ modules}
For any node $v \in V(G)$, we denote by $N(v)$ the set of all neighbors of $v$ in $G$. 
% Algorithm \LabA\/ and \MST\/ maintains the set $N(v)$ at every node $v$.
We use the following notations: 
\begin{itemize}
\item $\parent_v$: \textit{the parent} of $v$ in the current spanning tree, an integer pointer to a neighbor;%\vspace*{-0,3cm}
\item $\lab_v$: \textit{the label} of $v$ composed of a list of pairs of integers where each pair is an identifier and a distance (the size of $\lab_v$ is bounded by $O(\log^2 n)$ bits);%\vspace*{-0,3cm}
\item $\TS_v$: a pair of variables, the first one is an integer \textit{the number of nodes in the sub-tree} rooted at $v$ and the second one is the identifier of the child $u$ of $v$ with the maximum number of nodes in the sub-tree rooted at $u$;%\vspace*{-0,3cm}
\item $\m_v$: the \textit{minimum weighted edge} composed by a pair of variables, the first one is an integer, the weight of the edge and the second one is the label of a node $u$ stored in $\lab_u$.
\end{itemize}
%======================================================================

\subsection{Self-stabilizing Nearest Common Ancestor Labeling}
\label{sec:label}
Our labeling scheme, called in the following \LabA, uses the 
notions of \textit{heavy} and \textit{light} edges 
introduced in~\cite{HT84j}. 
In a tree, a \textit{heavy} edge is an edge between a node $u$ and one of its children $v$ of maximum number of nodes 
in its sub-tree. The other edges between $u$ and its other children are tagged as \textit{light} edges. We extend this edge designation to the nodes, a node $v$ is called \emph{heavy node} if the edge between $v$ and its parent is a heavy edge, otherwise $v$ is called \emph{light node}. Moreover, the root of a tree is a heavy node.
% The idea of the scheme is as follows. A tree is recursively divided into disjoint paths: the heavy and the light paths following the edges these paths use.
The idea of the scheme is as follows. A tree is recursively 
divided into disjoint paths: the heavy and the light paths which contain only heavy and light edges respectively.

\begin{figure}[ht]
\fbox{
\begin{minipage}{15cm}
\begin{description}
\item[-] $\Child \equiv \{u\in N(v): \parent_u=\id_v \}$
\item[-] $\SizeC(v) \equiv \Isleaf(v) \vee (\TS_v=(1+\sum_{u \in \Child(v)} \TS_u,\arg\max\{\TS_u: u\in \Child(v)\}) )$
\item[-] $\Isleaf(v) \equiv (\not \exists u \in N(v), \parent_u=\id_v)  \wedge \TS_v=(1,\bot)$
\item[-] $\Label(v) \equiv \LabelR(v) \vee \LabelNd(v)$	
\item[-] $\LabelR(v) \equiv (\parent_v=\emptyset \wedge \lab_v=(\id_v,0)) $
\item[-] $\LabelNd(v) \equiv  \parent_v \in N(v) \wedge ( \Heavy(v) \vee \Light(v))$
\item[-] $\Heavy(v) \equiv \TS_{\parent_v}[1]=\id_v \wedge \TS_v[0]< \TS_{\parent_v}[0] \wedge \last(\lab_{\parent_v})[1]+1=\last(\lab_v)[1]$
\item[-] $\Light(v) \equiv \TS_{\parent_v}[1]\neq \id_v  \wedge \TS_v[0]\leq \TS_{\parent_v}[0]/2 \wedge \lab_v=\lab_{\parent_v}.(\id_v,0)$\\
\item[-] $\Lca(u,v) \equiv \left\{ \begin{array}{llp{0,5cm}l} \lab.(a_0,a_1) & \mbox{ s.t. } \lab_u \cap \lab_v=\lab, \lab_u=\lab .(a_0,a_1).\lab_u' & & \mbox{\textbf{if} } (a_0=b_0 \vee \lab \neq \emptyset)\\ & \mbox{ and } \lab_v=\lab.(b_0,b_1).\lab_v' & & \hspace*{0,4cm} \wedge \lab_u \prec \lab_v \\ \lab.(b_0,b_1) & \mbox{ s.t. } \lab_u \cap \lab_v=\lab, \lab_u=\lab .(a_0,a_1).\lab_u' & & \mbox{\textbf{if} } (a_0=b_0 \vee \lab \neq \emptyset) \\ & \mbox{ and } \lab_v=\lab.(b_0,b_1).\lab_v' & & \hspace*{0,4cm} \wedge \lab_v \prec \lab_u \\ \emptyset & & &\mbox{\textbf{otherwise}}\end{array} \right.$
\item[-] $\Cycle(v) \equiv \lab_v \subset \lab_{\parent_v} \vee
\lab_v \prec \lab_{\parent_v} $
\item[-] $\IsMinEnabled(v) \equiv \Enabled(v) \wedge (\forall u \in N(v), \Enabled(u) \wedge \id_v < \id_u)$
\end{description}
\end{minipage}
}
\caption{Predicates used by the algorithm \LabA\/ for the labeling procedure.}
\label{fig:predicates1}
\end{figure}

To label the nodes in a tree $T$, the size of each subtree rooted at each node of $T$ is needed to identify heavy edges leading the heaviest subtrees at each level of $T$.
To this end, each node $v$ maintains 
a variable named $\TS_v$ which is a pair of integers. 
The first integer is the local 
estimation of the number of nodes in the subtree rooted at $v$. 
The second integer is the identifier of a child of $v$ with 
maximum number of nodes.  That is, it indicates the heavy edge.
The computation of $\TS_v$ is processed from the leaves to the root 
(see  Predicate $\SizeC(v)$ in Figure~\ref{fig:predicates1}).
A leaf has no child, therefore $\TS_v=(1,\bot)$ for a leaf node $v$ (see Predicate $\Isleaf(v)$ in 
Figure~\ref{fig:predicates1} and Rule $\RLC$).
%To summarize, he second element of $\TS$ indicates the heavy edge. 
%The number of nodes (variable $\TS$) is used to construct the 
%label and keep order the memory in $(\log^2 n)$. 

Based on the heavy and light nodes in a tree $T$ indicated by variable $\TS_v$ at each node $v \in T$, each node of $T$ can compute its label.
The label of a node $v$ stored in $\lab_v$ is a list of pair of integers.
Each pair of the list contains the identifier of a node and a distance to the root of the heavy path (i.e., a path including only heavy edges).
For the root $v$ of a fragment, the label $\lab_v$ is the following pair $(\id_v,0)$, respectively the identifier of $v$ and the distance to itself, i.e., zero (see Rule \RRoot\/).
When a node $u$ is tagged by its parent as a heavy node (i.e., $\TS_{\parent_v}[1]=\id_u$), then the node $u$ takes the label of its parent but 
it increases by one the distance of the last pair of the parent label.
Examples of theses cases are given in Figure~\ref{fig:label}, where integers inside the nodes are node identifiers and lists of pairs of values are node labels.
When a node $u$ is tagged by its parent $v$ as a light node (i.e., $\TS_{\parent_v}[1] \neq \id_u$), then the node $u$ becomes the root of a heavy path and it takes the label of its parent to which it adds a pair of integers composed of its identifier and a zero distance (see Figure~\ref{fig:label}). 

\begin{figure}[htbp]
\begin{center}
\includegraphics[scale=0.44]{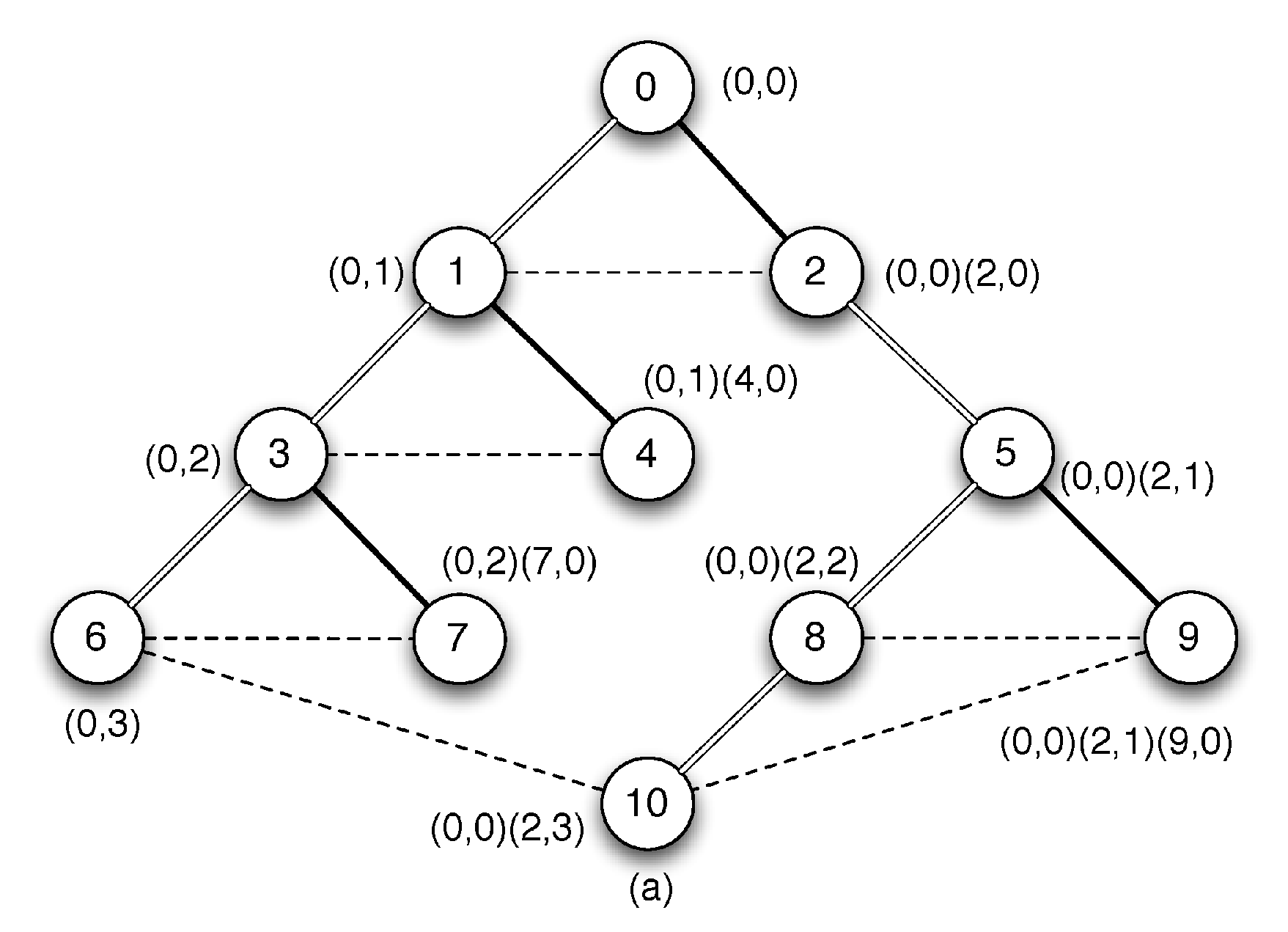}
\includegraphics[scale=0.44]{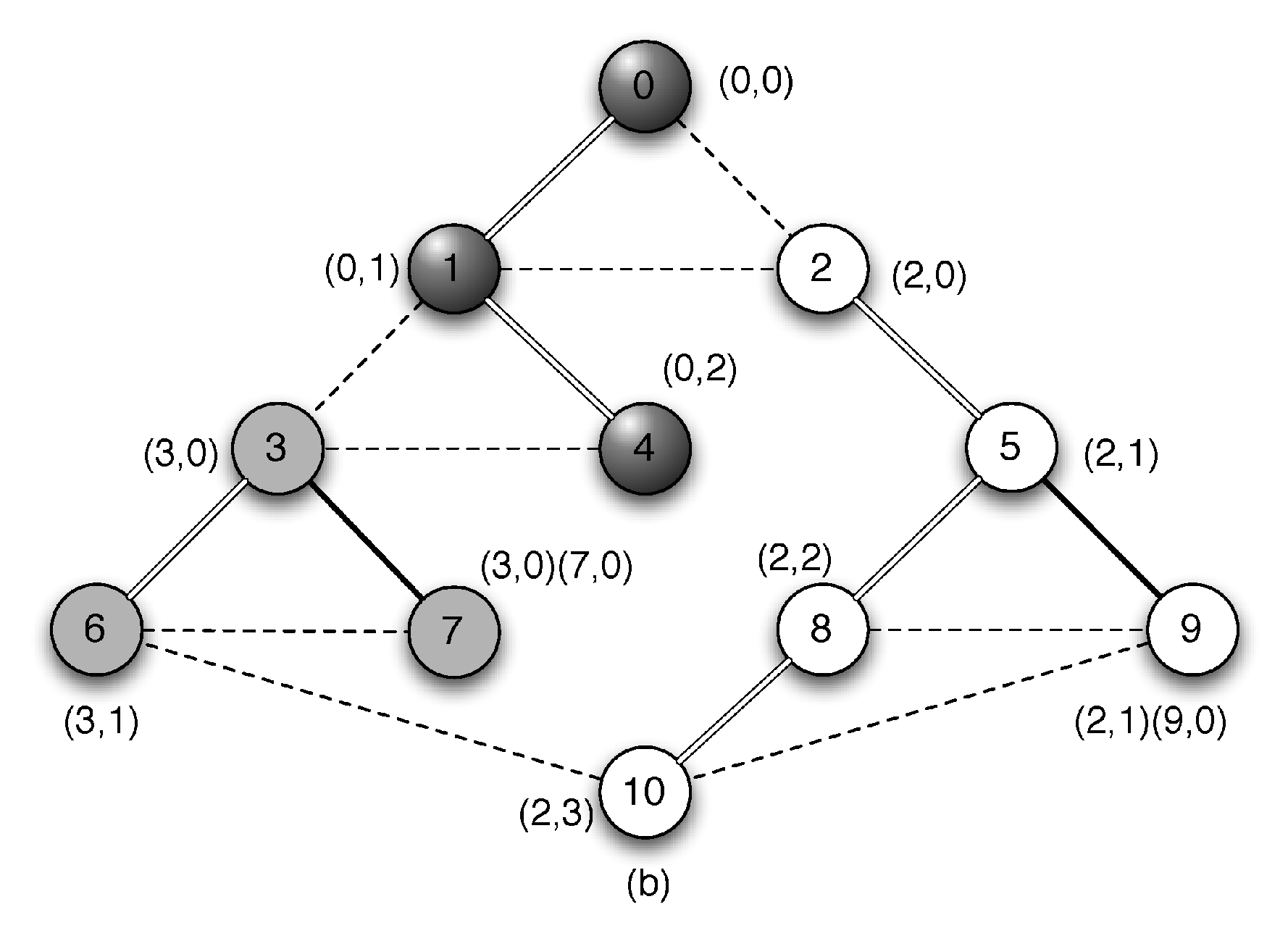}
\includegraphics[scale=0.65]{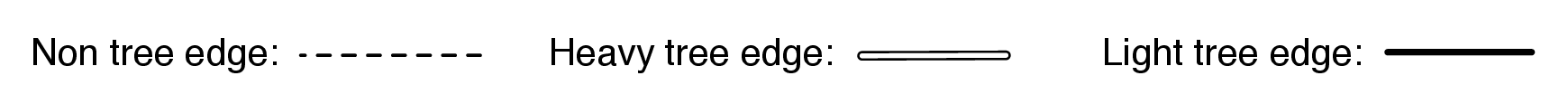}
\caption{Labeling scheme}
\label{fig:label}
\end{center}
\end{figure}
This labeling scheme is used in second part of this article in  \MST\/ algorithm to find the minimum weighted edges, but it is also used to detect and destroy cycles since the initial configuration may not be a spanning tree. To this end, we define an order $\prec$ on the 
labels of nodes. 
%part of their label with the same identifier then if the distance associated at this identifier is inferior for the node $a$ in comparison of the distance for the node $b$ then  $\lab_a <_{\lab} \lab_b$. More formally:\\
 Let $a$ and $b$ be two nodes and 
$\lab_a$ and $\lab_b$ be their respective labels 
such that $\lab_a=\lab .(a_0,a_1).\lab_a'$ and $\lab_b=\lab.(b_0,b_1).\lab_b'$ 
with $\lab_a \cap \lab_b=\lab$. %$\lab_a \prec \lab_b$ if (1) $a_0=\emptyset$ and $b_0\neq \emptyset$, or (2) $a_0<b_0$ or (3) $a_0=b_0$ and $a_1<b_1$.\\
The label of a node $a$ is lower than the label of node $b$, noted $\lab_a \prec \lab_b$, if (1) $(a_0,a_1).\lab_a'=\emptyset$ and $(b_0,b_1).\lab_b'\neq \emptyset$, or (2) $a_0<b_0$ or (3) $a_0=b_0$ and $a_1<b_1$.\\

A node $u$ can detect the presence of a cycle by only comparing its label with the label of its parent. That is, if its label 
is contained in the label of its parent, or it is inferior to the one of 
its parent then $u$ is part of a cycle (see Predicate $\Cycle(v)$ in Figure~\ref{fig:predicates1}). In this case, the node $u$ becomes the root of its fragment in order to break the cycle (see below Rule $\RRoot$).

 Algorithm \LabA\/ is composed of two rules. Rule \RRoot\/ 
creates a root or breaks cycles while rule \RLC\/ produces a proper 
labeling. Note that the last predicates in rules \RRoot\/ and \RLC\/ 
(the part in gray) are used only for insuring the exclusivity of rules execution when the labeling scheme works together with the \MST\/ scheme.

A node $v$ with an incoherent parent (which is not one of its neighbors) 
or present in a cycle executes \RRoot\/. Following the execution of this rule 
node $v$ becomes a root node, it sets its parent to void and its label to $(\id_v,0)$.\\

Rule \RLC\/ helps a node $v$ to compute the number of nodes in 
its sub-tree (stored in variable $\TS_v$) and provides to $v$ a coherent label. 
%This rule take allowance for the heavy edge ($\TS_{\parent_v}[1]=\id_v$) and light edge. 

%\begin{description}
%\item[$\RRoot$: [Root creation]] \\
%\textbf{If} $\parent_v \not \in N(v) \vee Cycle(v))$\\
%\textbf{Then} $\parent_v:=\emptyset \wedge \lab_v=(\id_v,0)$;
%\end{description}

\begin{description}
\item[$\RRoot$: [\ Root creation]] \\
\textbf{If} $\parent_v \not \in N(v) \vee (\parent_v=\emptyset \wedge \lab_v\neq(\id_v,0) \vee (\Cycle(v)$ \textcolor[gray]{0.5}{$\wedge \neg \NeedReorientation(v))$}\\
\textbf{Then} $\parent_v:=\emptyset;\ \lab_v=(\id_v,0)$;
\end{description}

%\begin{description}
%\item[$\RLC$: [Label correction]]\\
%\textbf{If} $\neg \Cycle(v) \wedge (\neg(\SizeC(v) \wedge \Label(v)) \vee \NP(v) )$\\
%\textbf{Then If} $\Isleaf(v)$ \textbf{then} $\TS_v=(1,\bot)$\\
%\hspace*{1cm} \textbf{Else } %$\TS_v=(\sum(\TS_u: u \in Child(v))+1,(\id_u: \max\{\TS_u: u \in Child(v)\})$\\
%$\TS_v:=(1+\sum_{u \in \Child(v)} \TS_u,\arg\max\{\TS_u: u\in \Child(v)\});$\\
%\hspace*{1cm} \textbf{If} % $\TS_{\parent_v}[1]=\id_v$ \textbf{then} $\lab_v=\lab_{\parent_v}, e=\last(\lab_v), \lab_v=\lab_{\parent_v}\setminus e, \lab_v=\lab_v.(e[0],e[1]+1)$\\
%$\TS_{\parent_v}[1]=\id_v \mbox{\textbf{ then }} \lab_v:=\lab_{\parent_v};\ last(\lab_v)[1]:=last(\lab_v)[1]+1;$\\
%\hspace*{1,1cm}\textbf{Else} $\lab_v=\lab_{\parent_v}.(\id_v,0)$\\
%\end{description}

\begin{description}
\item[$\RLC$: [\ Label correction]]\\
\textbf{If} $\neg \Cycle(v)  \wedge (\neg \SizeC(v) \vee \neg \Label(v))$ \textcolor[gray]{0.5}{$\wedge \IsMinEnabled(v) \wedge \neg \TreeChange(v)$}\\
\textbf{Then If} $\Isleaf(v)$ \textbf{then} $\TS_v=(1,\bot)$\\
\hspace*{1cm} \textbf{Else } %$\TS_v=(\sum(\TS_u: u \in Child(v))+1,(\id_u: \max\{\TS_u: u \in Child(v)\})$\\
$\TS_v:=(1+\sum_{u \in \Child(v)} \TS_u,\arg\max\{\TS_u: u\in \Child(v)\});$\\
\hspace*{1cm} \textbf{If} % $\TS_{\parent_v}[1]=\id_v$ \textbf{then} $\lab_v=\lab_{\parent_v}, e=\last(\lab_v), \lab_v=\lab_{\parent_v}\setminus e, \lab_v=\lab_v.(e[0],e[1]+1)$\\
$\TS_{\parent_v}[1]=\id_v \mbox{\textbf{ then }} \lab_v:=\lab_{\parent_v};\ last(\lab_v)[1]:=last(\lab_v)[1]+1;$\\
\hspace*{1,1cm}\textbf{Else} $\lab_v=\lab_{\parent_v}.(\id_v,0)$
\end{description}

\subsection{Self-stabilizing \MST}

In this section we describe our self-stabilizing \MST\/
algorithm. The algorithm executes two phases: the MST correction and
the MST fragments merging. Recall that our algorithm uses the blue rule to
construct a spanning tree and the red rule to recover from invalid configurations. 
In both cases, it uses the nearest-common ancestor labeling scheme to
identify fragments and fundamental cycles.
We assume in the following that the \textit{merging} operations have a higher priority
than the \textit{recovering} operations. That is, the system recovers
from an invalid configuration if and only if no merging operation is 
possible. In the worst case, after a failure hit the system, a
merging phase will be followed by a recovering phase and finally by a
final merging phase.

\subsubsection{The minimum weighted edge and MST correction}

\begin{figure}[!ht]
\fbox{
\begin{minipage}{15cm}
\begin{description}
%\item[-] $\NP(v) \equiv \lab_v=(\bot,\bot) \wedge \m_{\parent_v}=(\emptyset,\emptyset)$	\vspace*{-0,3cm}
\item[-] $\FC(v) \equiv \min\{ \m_u:u\in \Child(v) \wedge \m_u[1]=\emptyset\}$%\vspace*{-0,3cm}
\item[-] $\FA(v) \equiv (\min\{w(u,v):u\in N(v)\setminus \Child(v) \setminus \{\parent_v\}\wedge \Lca(u,v)= \emptyset\},\emptyset)$%\vspace*{-0,3cm}
\item[-] $\Fusion(v) \equiv \min\{\FC(v),\FA(v)\}$
\item[-] $\FarLcaC(v) \equiv \arg\min_{\prec}\{ \m_u[1] : u \in \Child(v) \wedge \m_u[1] \preceq \lab_v \}$%\vspace*{-0,3cm}
\item[-] $\UC(v) \equiv \m_u \mbox{ such that } \FarLcaC(v)=u$%\vspace*{-0,3cm}
\item[-] $\FarLca(v) \equiv \arg\min_{\prec}\{\Lca(u,v) : u\in
N(v)\setminus \Child(v) \setminus \{\parent_v\} \wedge \m_v[1] \neq
\emptyset \wedge$ \\
\hspace*{8cm}$\Lca(u,v) \succ \m_v[1] \wedge \Lca(u,v)\neq \emptyset\}$\\
\item[-] $\UA(v) \equiv \left\{ \begin{array}{ll} (w(u,v),\Lca(u,v)) \mbox{ s.t. } \FarLca(v)=u & \mbox{\textbf{if }} \FarLca(v) \neq \emptyset \\ (w(u,v),\Lca(u,v)) & \mbox{\textbf{otherwise}} \\ \ \mbox{s.t. } u=\arg\min_{\prec}\{\Lca(u,v): & \\ \ u \in N(v) \setminus Child(v) \setminus \{\parent_v\} \wedge \Lca(u,v) \neq \emptyset\} & \end{array} \right.$%\vspace*{-0,3cm}
\item[-] $\Update(v) \equiv \left\{ \begin{array}{ll} \UC(v) & \mbox{\textbf{if }} \UA(v)[1] \prec \UC(v) \\ \UA(v) & \mbox{\textbf{otherwise}} \end{array} \right.$
\item[-] $\MinEdge(v) \equiv \left\{ \begin{array}{ll} \Fusion(v) & \mbox{\textbf{if} } \Fusion(v) \neq \emptyset\\ \Update(v) & \mbox{\textbf{otherwise}} \end{array} \right.$

\item[-] $\NeedReorientation(v) \equiv \LabelR(v) \vee (\lab_{\parent_v}=(\bot,\bot) \wedge \parent_{\parent_v}=\id_v)$%\vspace*{-0,3cm}
\item[-] $\EndReorientation(v) \equiv \lab_{\parent_v}=(\emptyset,\emptyset) \wedge (\lab_v=(\bot,\bot) \vee \lab_v \neq (\emptyset,\emptyset))$%\vspace*{-0,3cm}
\item[-] $\TreeChange(v) \equiv \NeedReorientation(v) \vee \EndReorientation(v)$

\item[-] $\NewFragment(v) \equiv \m_v=\m_{\parent_v} \wedge \m_v[1] \neq \id_v \wedge w(v,\parent_v)>\m_v[0]$
\end{description}
\end{minipage}
}
\caption{Predicates used by the \MST\/ for the tree correction or the fusion fragments.}
\label{fig:predicates2}
\end{figure}

Note that the scope of our labeling scheme is twofold. First, it
allows a node to identify the neighbors that share the same fragment
and consequently to select the outgoing edges of a fragment. 
Second, the labeling scheme may be used to identify cycles and
to repair the tree. To this end, the algorithm uses the nearest common
ancestor predicate  \Lca\/ depicted 
in Figure~\ref{fig:predicates1}. For two nodes $u$ and $v$ with
$e=(u,v)$ a non tree edge (i.e., $\parent_u\neq v$ and
$\parent_v\neq u$), if the nearest common ancestor does not exist then
$u$ and $v$ are in two distinct fragments (i.e., if we have $\Lca(u,v)=\emptyset$). Otherwise $u$ and $v$ are in the
same fragment $F$ and the addition of $e$ to $F$ generates a
cycle. Let $\path(x,y)$ be the set of edges on the
unique path between $x$ and $y$ in 
$F$, with $x,y\in F$. The fundamental cycle 
$C_e$ is the following: $C_e=\path(u,\Lca(u,v))\cup \path(\Lca(u,v),v)\cup e$. 
Consider the example depicted on  Figure~\ref{fig:label}(b).
The labels of nodes $10$ and $6$ are respectively  $\lab_{10}=(2,3)$ and
$\lab_{6}=(3,1)$. In this case $\Lca(10,6)=\emptyset$ so the edge
$(10,6)$ is an outgoing edge because the nodes $10$ and $6$ are in two distinct fragments and they have no common ancestor. If the edge $(10,6)$ is of minimum weight then
$(10,6)$ can be used for a merging 
between the fragment rooted in $2$ and the fragment rooted in $3$. 
For the case of nodes $10$ and $9$ the labels are $\lab_{10}=(2,3)$
and $\lab_{9}=(2,1)(9,0)$ and 
$\Lca(10,9)=(2,1)$. Consequently, $10$ and $9$ are in the same fragment. 
The fundamental cycle $C_e$ with $e=(9,10)$ goes through the node with
the label $\Lca(10,9)$, in other word 
the node $5$ in Figure~\ref{fig:label}(b).

%For this purpose we use the variable  $\m=(u,v)$, which represent the
%minimum outgoing edge for a fragment or the edge for a fundamental
%cycle. 
% The core of the computation of the minimum weight edge is based on Predicate $\MinEdge(v)$ (see Figure~\ref{fig:predicates2}). 
Predicate $\MinEdge(v)$ (see Figure~\ref{fig:predicates2}) computes both the minimum weight outgoing edge used in a merging phase and the internal edges used in a recovering phase. Our algorithm gives priority to the computation of minimum outgoing edges via Predicate $\MinEdge(v)$. A recovering phase is initiated if there exists a unique tree or if a sub-tree of one fragment has no outgoing edge.

%  The information about the weight of the non tree edges are sent from the leaves to the root.
%  and Rule \RMin\/

The computation of the minimum weight outgoing edge is done in a fragment $F_u$ if and only an adjacent fragment $F_v$ is detected by $F_u$, i.e., if we have Predicate $\Fusion(v) \neq \emptyset$. In this case, using Rule $\RMin$ each node collects from the leaves to the root the outgoing edges leading to an adjacent fragment $F_v$. At each level in a fragment, a node selects the outgoing edge of minimum weight among the outgoing edges selected by its children and its adjacent outgoing edges. Thus, this allows to the root of a fragment to select the minimum outgoing edge $e$ of the fragment leading to an adjacent fragment. Then, the edge $e$ can be used to perform a merging between two adjacent fragments using an edge belonging to a MST.

% For a merging phase, our algorithm gives priority to the outgoing-edges of minimum weight
% via Predicate $\MinEdge(v)$.
% Namely, an outgoing edge $e=(u,v)$ is an edge such that $\Lca(u,v)=\emptyset$.
% Therefore for the tree recovering, the weighted edges are sent only
% if there exists a unique tree or 
% if a sub-tree of one fragment has no outgoing edge. 

Let us explain Rule $\RC$ which allows to correct a tree (or a fragment). 
In this case, the information about the non-tree edges are sent to the
root as follows. 
Among all its non-tree edges, a node $u$ sends the edge $e=(u,v)$ with the $\Lca(u,v)$ nearest to the root (see Figure~\ref{fig:min}(a)).
The information about the edge $e$ is stored in variable $\m_u$.
If the parent $x$ of the node $u$ has the same information and  the
weight of the edge $w(u,x)>w(e)$ then 
the edge $(u,v)$ is removed from the tree (see Figure~\ref{fig:min}(a-b) for the nodes 6 and 10).
% In the same way, if $x$ has an information for a non tree edge $e'$ whose nearest common ancestor is closest to the root than the nearest common ancestor $\Lca(u,v)$, i.e., $\Lca(e') \prec \Lca(u,v)$, then $(u,v)$ is removed from the tree too (see Figure~\ref{fig:min}(a-b) for the nodes 7 and 4 ).
We use the red rule in an intensive way, because we remove all the edges with a weight upper than $w(e)$ in fundamental cycle of $e$. 
This interpretation of the red rule allows to insure that after a recovering phase the remaining edges belong to a MST.

\begin{figure}[htbp]
\begin{center}
\includegraphics[scale=0.4]{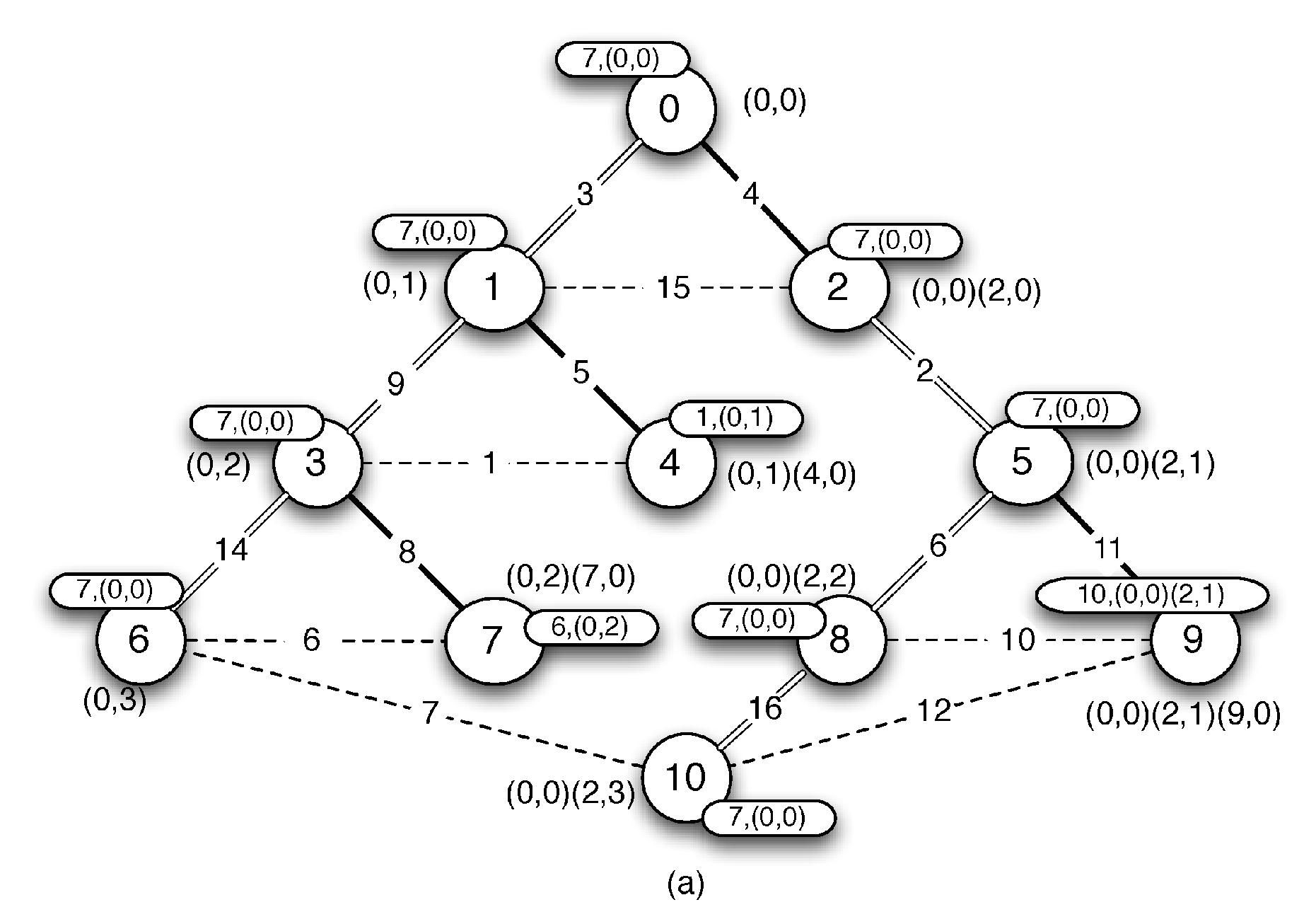}
\includegraphics[scale=0.4]{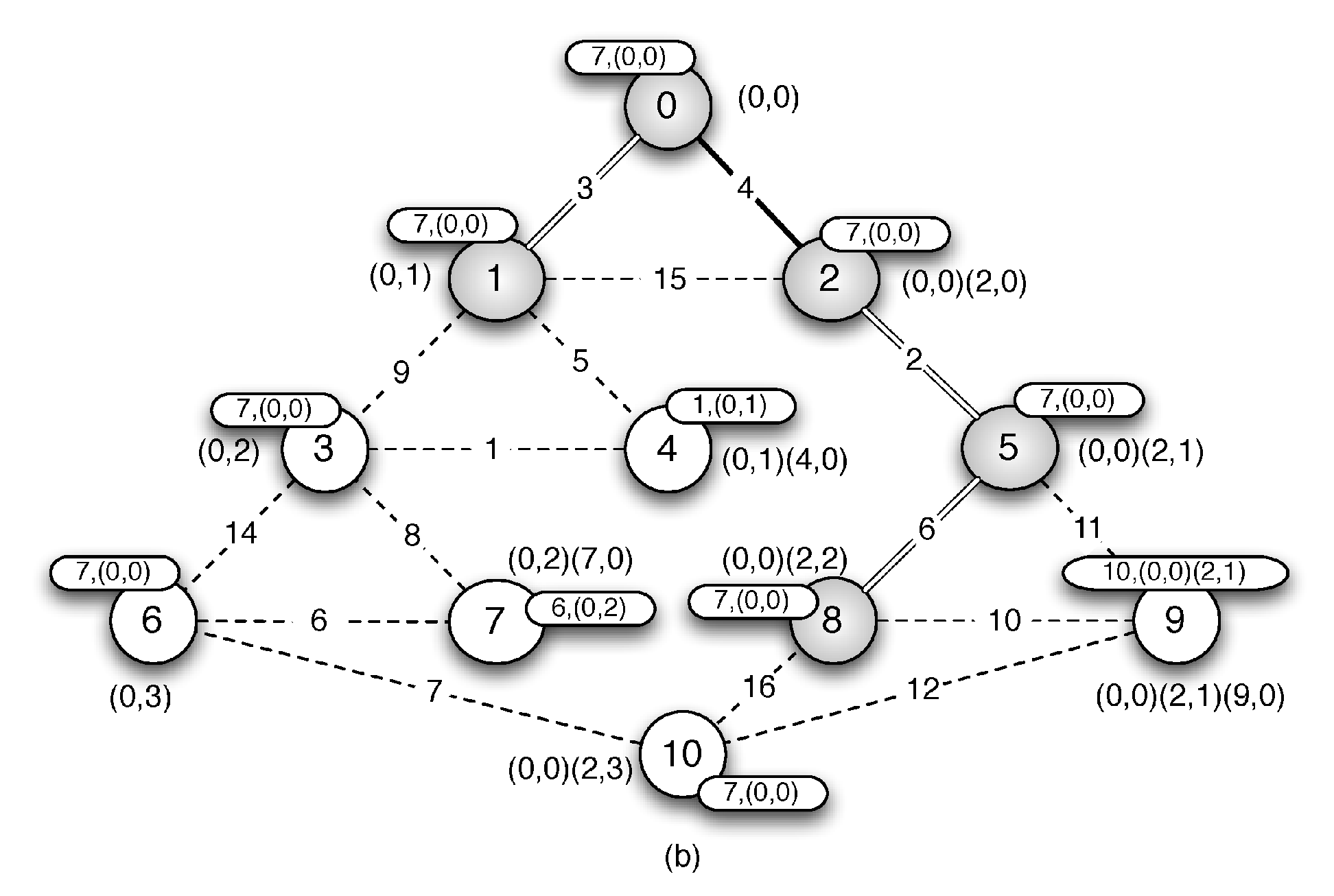}
\includegraphics[scale=0.65]{Comments.ps}
\caption{Minimum weighted edge computation and Tree correction. The bubble at each node $v$ corresponds to the weight and the label of the common ancestor of the edge stored on variable $\m_v$.}
\label{fig:min}
\end{center}
\end{figure}

%\begin{description}
%\item $\RMin$: \textbf{Minimum computation} \\
%\textbf{If} $\neg \Cycle(v) \wedge \Label(v)  \wedge (\m_v\neq \Fusion(v)\wedge \m_v\neq\Update(v))$\\
%\textbf{Then If } $\Fusion(v)\neq \emptyset$\\
%\hspace*{1,5cm} \textbf{Then }$ \m_v=\Fusion(v)$\\
%\hspace*{1,2cm}\textbf{Else If } $ \m_v=\m_{\parent{v}}$ \textbf{Then }\\
%\hspace*{3cm}\textbf{If } $w(v,\parent_v)>\m_v[0]$\\
%\hspace*{3,5cm}\textbf{Then } $\parent_v:=\emptyset ; \lab_v:=(\id_v,0)$;\\
%\hspace*{3cm} \textbf{Else }$ \m_v:=\Update(v)$\\
%\end{description}

% \begin{description}
% \item $\RMin$: [ \textbf{Minimum computation} ] \\
% \textbf{If} $\neg \Cycle(v) \wedge \Label(v)  \wedge \m_v \neq \MinEdge(v)$\\
% \textbf{Then } $\m_v:=\MinEdge(v);$\\
% \end{description}

\begin{description}
\item $\RMin$: [ \textbf{Minimum computation} ] \\
\textbf{If} $\neg \Cycle(v) \wedge \Label(v)  \wedge [(\m_v \neq \MinEdge(v) \wedge \Fusion(v)\neq \emptyset) \vee (\m_{\parent_v}=\m_v \wedge \Fusion(v)=\emptyset)]$\\
\textbf{Then } $\m_v:=\MinEdge(v);$
\end{description}

% \begin{description}
% \item $\RC$: [ \textbf{MST Correction} ] \\
% \textbf{If} $\neg \Cycle(v) \wedge \Label(v)  \wedge \m_v = \Update(v) \wedge \NewFragment(v)$\\
% \textbf{Then } $\parent_v:=\emptyset;\ \lab_v:=(\id_v,0);$\\
% \end{description}

\begin{description}
\item $\RC$: [ \textbf{MST Correction} ] \\
\textbf{If} $\neg \Cycle(v) \wedge \Label(v) \wedge \NewFragment(v)$\\
\textbf{Then } $\parent_v:=\emptyset;\ \lab_v:=(\id_v,0);$
\end{description}

To summarize, in this section we explained how to compute the outgoing-edges and the
fundamental cycles (Rule \RMin\/),  and how to recover from a false tree
(Rule \RC\/). The next section addresses the fragments merging operation (Rules \RF\/ and \REnd).

\subsubsection{Fragments merging}

In this phase two rules are executed: $\RF$ and $\REnd$. 
Note that Rule \RMin\/ (described in the previous section) computes
from the leaves to the root 
the minimum outgoing edge $e=(u,v)$ of the fragment $F_u$, with $u\in F_u$.
The information about $e$ are stored in the variable \m\/, i.e., the weight of the edge and a common ancestor equal to $\infty$ to indicate that these information concern an outgoing edge.  
When a root $r$ of $F_u$ has stabilized its variable $\m_r$, it starts a merging phase (Rule \RF\/).
To this end, the nodes in the path between $r$ and $u$ are reoriented from $r$ to $v$.
During this reorientation the labels are locked. That is, each node $x$ on the path between $r$ and $u$ (including $r$ and excluding $v$) changes its label to: $\lab_v:=(\bot,\bot)$.
When a node $u$ becomes the root of the fragment $F_u$ it can merge
with the fragment $F_v$.
%, finally $v$ becomes the parent of $v$.
After the addition of the outgoing edge $e$, the labeling process is re-started (see Rule \REnd).
The merging phase is repeated until a single fragment is obtained.
%only one fragment remain, the final MST.

%\begin{description}
%\item $\RF$: \textbf{Merging} \\
%\textbf{If} $\Label(v) \wedge \m_v=\Fusion(v)$\\
%\textbf{Then If}$( \exists u \in Child(v): \m_u=\m_v)$\\
%\hspace*{1,6cm}\textbf{Then } $\parent_v:=\min\{\id_u: u \in Child(v) \wedge \m_u=\Fusion(v)\};\ \lab_v:=(\bot,\bot)$;\\
%\hspace*{1,2cm}\textbf{Else If }$(\exists u \in N(v): \m_v[0]=w(u,v))$\\
%\hspace*{1,6cm}\textbf{Then} $\parent_v:=\min\{\id_u: u \in N(v)\backslash Child(v) \wedge \m_v[0]=w(u,v)\};\  \m_v:=(\emptyset,\emptyset)$;
%\end{description}

% \begin{description}
% \item $\RF$: [ \textbf{Merging} ] \\
% \textbf{If} $\NeedReorientation(v) \wedge \m_v=\Fusion(v)$\\
% \textbf{Then If } $(\exists u \in N(v) \backslash Child(v) \backslash \{\parent_v\}, w(u,v)=\Fusion(v) \wedge \id_v>\id_u)$\\
% \hspace*{1,2cm}\textbf{Then } $\parent_v:=\min\{\id_u: u \in N(v) \backslash Child(v) \backslash \{\parent_v\} \wedge w(u,v)=\Fusion(v)\};$\\
% \hspace*{2,35cm} $\lab_v:=(\emptyset,\emptyset)$;\\
% \hspace*{1,2cm}\textbf{If } $(\exists u \in N(v) \backslash Child(v) \backslash \{\parent_v\}, w(u,v)=\Fusion(v) \wedge \id_v<\id_u \wedge \parent_u=\id_v)$\\
% \hspace*{1,2cm}\textbf{Then } $\lab_v:=(\emptyset,\emptyset)$;\\
% \hspace*{1,2cm}\textbf{Else } $\parent_v:=\min\{\id_u: u \in Child(v) \wedge \m_v=\Fusion(v)\};\ \lab_v:=(\bot,\bot)$;
% \end{description}

\begin{description}
\item $\RF$: [ \textbf{Merging} ] \\
\textbf{If} $\NeedReorientation(v) \wedge \m_v=\Fusion(v)$\\
\textbf{Then}\\
\hspace*{0,5cm}\textbf{If } $(\exists u \in N(v) \backslash Child(v) \backslash \{\parent_v\}, w(u,v)=\Fusion(v) \wedge \id_v>\id_u)$\\
\hspace*{0,5cm}\textbf{Then }\\
\hspace*{0,9cm} $\parent_v:=\min\{\id_u: u \in N(v) \backslash Child(v) \backslash \{\parent_v\} \wedge w(u,v)=\Fusion(v)\};$\\
\hspace*{0,9cm} $\lab_v:=(\emptyset,\emptyset)$;\\
\hspace*{0,5cm}\textbf{If } $(\exists u \in N(v) \backslash Child(v) \backslash \{\parent_v\}, w(u,v)=\Fusion(v) \wedge \id_v<\id_u \wedge \parent_u=\id_v)$\\
\hspace*{0,5cm}\textbf{Then } $\lab_v:=(\emptyset,\emptyset)$;\\
\hspace*{0,5cm}\textbf{Else } $\parent_v:=\min\{\id_u: u \in Child(v) \wedge \m_v=\Fusion(v)\};$\\
\hspace*{1,4cm} $\lab_v:=(\bot,\bot)$;
\end{description}

% \begin{description}
% \item $\RF$: [ \textbf{Merging} ] \\
% \textbf{If} $\NeedReorientation(v) \wedge \m_v=\Fusion(v)$\\
% \textbf{Then If } $(\exists u \in N(v) \backslash Child(v) \backslash \{\parent_v\}, w(u,v)=\Fusion(v) \wedge \id_v>\id_u)$\\
% \hspace*{1,2cm}\textbf{Then } $\parent_v:=\min\{\id_u: u \in N(v) \backslash Child(v) \backslash \{\parent_v\} \wedge w(u,v)=\Fusion(v)\};$\\
% \hspace*{2,35cm} $\lab_v:=(\emptyset,\emptyset)$;\\
% \hspace*{1,2cm}\textbf{Else } $\parent_v:=\min\{\id_u: u \in Child(v) \wedge \m_v=\Fusion(v)\};\ \lab_v:=(\bot,\bot)$;
% \end{description}

\begin{description}
\item $\REnd$: [ \textbf{End Merging} ] \\
\textbf{If} $\neg \NeedReorientation(v) \wedge \EndReorientation(v)$\\
\textbf{Then} $\lab_v:=(\emptyset,\emptyset)$;
\end{description}

%\begin{figure}[htbp]
%\begin{center}
%\includegraphics[scale=0.5]{LabelExFragW}
%\caption{default}
%\label{default}
%\end{center}
%\end{figure}

%\subparagraph{\hspace*{-0,5cm}Concluding remarks of the \MST\/ algorithm}
%The algorithm finishes when all non-tree edges are in a same fragment, the computation of the minimum weighted edge for all nodes becomes stable and none of the nodes change its parent. Then the MST is computed and none variables change if not corruption occur.  Our algorithm is silent. 

%======================
%Proofs

\section{Correctness proof}

\begin{lemma}
\label{lem:no_cycle}
Let $\mathcal{C}$ a configuration where the set of variables
$\parent_v, v \in V,$ form at least one cycle in the network. 
In a finite time, Algorithm \LabA\/ removes all the cycles from the network.
\end{lemma}

\begin{proof}
If a node $v$ has a parent which is not in its neighborhood or if $v$ has no parent then the parent and the label variable of $v$ is modified to $\emptyset$ and $(\id_v,0)$ respectively with Rule $\RRoot$.

A node $v$ identifies a cycle with Predicate $\Cycle(v)$ which uses $v$'s label and the label of its parent. In a legitimate configuration, $v$'s label is smaller than the label of its parent and is constructed using the label of its parent, i.e., the label of the parent of $v$ is included to $v$'s label. Thus, if the label of $v$ is included or is smaller than the label of its parent then a cycle is detected and Predicate $\Cycle(v)$ is true. In this case, $v$ reinitiates its parent and label variable using Rule $\RRoot$.

In order to detect a cycle the label computation process must cross a part or all the nodes of the cycle. However, since we consider a distributed scheduler then all the nodes in a cycle can be activated and we can have a rotation of the labels of the nodes in the cycle. This may lead to a new configuration in which the labels cannot be used to detect a cycle, because the label of one node is not used to compute some other labels and to detect a cycle. To break this symmetry, we use the node identifiers with Predicate $\IsMinEnabled(v)$. This predicate allows to activate the node $v$ iff $v$ has no neighbor $u$ such that $u$ is activated and $u$'s identifier is lower than $v$. Therefore, there is at least a node $x$ in the cycle which is not activated and when the label of $x$ is used by some other nodes to compute their labels then Predicate $\Cycle(x)$ is true and $x$ breaks the cycle using Rule $\RRoot$.
\end{proof}

$\newline$
According to Lemma~\ref{lem:no_cycle}, if the system starts from a configuration which contains at least one cycle then all the cycles are removed from the network in a finite time. Therefore, in the following we consider only configurations containing no cycle.

\begin{definition}[Legitimate state of $\LabA$]
\label{def:label_legitimate_configuration}
Let $\mathcal{C}$ a configuration with no cycle in the network, i.e., which contains a forest of trees $T=(V_T \subseteq V,E_T \subseteq E)$. The configuration $\mathcal{C}$ is legitimate for Algorithm \LabA\/ iff each node $v \in V_T$ satisfies one of the following conditions:
\begin{enumerate}
\item the label $\lab_v$ of $v$ is equal to $\lab_{\parent_v}.(\id_v,0)$, if the edge between $v$ and $v$'s parent in $T$ is a light edge or $v$ is the root of the tree;
\item the label $\lab_v$ of $v$ is equal to $\lab_{\parent_v}$ and $last(\lab_v)[1]=last(\lab_{parent})[1]+1$, if the edge between $v$ and $v$'s parent in $T$ is a heavy edge.
\end{enumerate}
\end{definition}

\begin{lemma}[Convergence for $\LabA$]
\label{lem:label_convergence}
Starting from an illegitimate configuration for Algorithm $\LabA$, eventually Algorithm \LabA\/ reaches in a finite time a legitimate configuration.
\end{lemma}

\begin{proof}
To compute correct node labels, heavy and light edges in each tree $T=(V_T,E_T)$ of the forest in the network must be identified. To this end, each node $v \in V_T$ maintains in the variable $\TS_v$ two information: the size of its subtree in $T$ and the identifier of the child with the subtree of maximum number of nodes. Based on the first information given by its children $u \in V_T$ stored in variable $\TS_u$, each node $v$ compute the size of its subtree in $T$ and informs its child $u$ if the edge $(u,v) \in E_T$ is a light or heavy edge. The edge $(u,v) \in E_T$ is a heavy edge if the second information stored in $\TS_v$ is equal to $\id_u$ the identifier of $u$, a light edge otherwise. Therefore, each node $v \in V_T$ can detect if its label is correct according to its parent label. Note that Predicate $\IsMinEnabled(v)$ is used at node $v \in V$ to help to break cycle if we are in the case of a configuration described in proof of Lemma~\ref{lem:no_cycle}. Moreover, since we use the labeling scheme with minimum spanning tree computation rules Predicate $\TreeChange(v)$ is used to forbid the label correction when it has been modified by Rule $\RC, \RF$ and $\REnd$.

The computation of the information stored in the variable $\TS_v$ at each node $v \in V_T$ is done via bottom-up fashion in the tree $T$. According to Predicate $\SizeC(v)$, if the variable $\TS_v$ is not equal to $(1,\bot)$ at a leaf node $v$ in $T$ then Predicate $\SizeC(v)=false$ and $v$ can execute Rule $\RLC$ to correct its variable $\TS_v$. Otherwise according to Predicate $\SizeC(v)$, for any internal node $v \in V_T$ the first information of $\TS_v$ must be equal to one plus the sum of the size of the children subtrees and the second one to the identifier of its child with the maximum subtree size. Thus, if variable $\TS_v$ is not correct (i.e., Predicate $\SizeC(v)=false$) then $v$ can execute Rule $\RLC$ to correct its variable $\TS_v$. Using the same argument, one can show by induction that for any internal node $v \in V_T$ we have $\TS_v=(1+\sum_{u \in \Child(v)} \TS_u,\arg\max\{\TS_u: u\in \Child(v)\})$.

The computation of the node labels in a tree $T$ is done via top-down fashion starting from the root of $T$. For convenience, a path is called \emph{heavy} (resp. \emph{light}) if it contains only heavy (resp. light) edges. Moreover, a node is called \emph{heavy} (resp. \emph{light}) if the edge between its parent and itself is a heavy (resp. light) edge. $v$ is informed by its parent with $\TS_{\parent_v}$ if $v$ is a heavy (i.e., $\TS_{\parent_v}[1]=\id_v$) or light (i.e., $\TS_{\parent_v}[1] \neq \id_v$) node. The root node $v$ of $T$ is also the root of a heavy path and its label must be equal to $(\id_v,0)$. According to Predicate $\Label(v)$ and $\LabelR(v)$, $v$ can execute Rule $\RLC$ to correct its label. Otherwise, we have two cases: heavy or light nodes. When the root have a correct label then all its children can compute their correct label. If a heavy (resp. light) node has a label different from $\lab_{\parent_v}$ and $last(\lab_v)[1]=last(\lab_{parent})[1]+1$ (resp. $\lab_{\parent_v}.(\id_v,0)$) then we have Predicate $Heavy(v)=false$ (resp. $\Light(v)=false$), $\LabelNd(v)=false$ and $\Label(v)=false$. Therefore, $v$ can execute Rule $\RLC$ to correct its label $\lab_v$ accordingly with its parent. Using the same argument, one can show by induction that for any internal node $v \in V_T$ we have $\lab_{\parent_v}$ and $last(\lab_v)[1]=last(\lab_{parent})[1]+1$ (resp. $\lab_{\parent_v}.(\id_v,0)$) for heavy (resp. light) nodes.
\end{proof}

\begin{lemma}[Closure for $\LabA$]
\label{lem:label_closure}
The set of legitimate configurations for \LabA\/ is closed.
\end{lemma}

\begin{proof}
According to Algorithm $\LabA$, the labeling procedure is done using only Rule $\RLC$. In any legitimate configuration for Algorithm $\LabA$, for any node $v \in V$ Predicate $\SizeC(v)$ and $\Label(v)$ are true and Rule $\RLC$ cannot be executed by a node $v$. So, starting from a legitimate configuration for Algorithm $\LabA$ the system remains in a legitimate configuration.
\end{proof}

\begin{definition}[Legitimate state of $\MST$]
\label{def:mst_legitimate_configuration}
A configuration is legitimate for Algorithm \MST\/ iff each node $v \in V$ satisfies the following conditions:
\begin{enumerate}
\item a tree $T$ spanning the set of nodes in $V$ is constructed;
\item $T$ is of minimum weight among all spanning trees.
\end{enumerate}
\end{definition}

% - remonté arête de plus petit poids sortant du fragment
% - remonté arête de correction du plus loin au plus proche de la racine
% - si plus de fusion possible et fragment non correct, alors fragment est corrigé
% - arête supprimée par règle correction MST ne peut être choisie par règle fusion (si pas de changement de poids)
% - fusion réalisé via arête de plus petit poids sortant du fragment
% -> Debut fusion : inversion correcte de l'orientation jusqu'à arête de fusion + changement label
% -> Fin fusion : changement label pour informer fin fusion

\begin{lemma}
\label{lem:mwoe_computation}
Let $T_i=(V_{T_i},E_{T_i})$ a tree (or fragment). Eventually for each node $v \in V_{T_i}$ the variable $\m_v$ contains a pair of values: the weight of the minimum outgoing edge $(u,v)$ of the fragment $T_i$ and $\emptyset$, if a merging is possible between two fragments $T_j$ and $T_i$, $(j\neq i)$.
\end{lemma}

\begin{proof}
We assume that $T_i$ and $T_j$, $(j\neq i)$ are two distinct coherent trees (i.e., different root and correct labels, otherwise Rule $\RRoot$ and $\RLC$ are used to correct the trees) in the network and that a merging is possible between $T_i$ and $T_j$. The computation of the minimum outgoing edge of $T_i$ (resp. $T_j$) is done in a bottom-up fashion. We consider the tree $T_i$ but the computation in $T_j$ is done in a same way. A leaf node $v$ can compute and store in its variable $\m_v$ its local adjacent minimum outgoing edge leading to another tree. Macro $\MinEdge(v)$ returns the local minimum outgoing edge if Macro $\Fusion(v) \neq \emptyset$. To this end, if there is an adjacent outgoing edge $(u,v)$ leading to another tree $T_j$ (i.e., $\FA(v) \neq \emptyset$ and $\Fusion(v) \neq \emptyset$) and we have $\m_v \neq \MinEdge(v)$ then $v$ can execute Rule $\RMin$ to compute its local outgoing edge stored in the variable $\m_v$. A internal node $v$ must use the local outgoing edges computed by its children and selects the edge of minimum weight among these edges, then it compares this value with the weight of its adjacent local outgoing edge and again it holds the edge of minimum weight. The selection of its children minimum outgoing edge is done by Macro $\FC(v)$ and the computation of its local outgoing edge is done by Macro $\FA(v)$ as for a leaf node. Thus, if there is an outgoing edge which can be used to make a merging with another tree (i.e., $\Fusion(v) \neq \emptyset$) and we have $\m_v \neq \MinEdge(v)$ for a internal node $v$, then $v$ can execute Rule $\RMin$ to compute in the variable $\m_v$ its local minimum outgoing edge. Using the same argument, one can show by induction that for any internal node $v \in V_{T_i}$ we have $\m_v=\Fusion$. Therefore, the local outgoing edge computed by the root node $v$ is the minimum outgoing edge of $T_i$.
\end{proof}

\begin{lemma}
\label{lem:recover_computation}
Let $T_i=(V_{T_i},E_{T_i})$ a tree (or fragment). Eventually for each node $v \in V_{T_i}$ the variable $\m_v$ contains a pair of values: the weight of an edge $(u,v) \not \in E_{T_i}$ and the label of the nearest common ancestor of $u \in V_{T_i}$ and $v$. Moreover, eventually all local internal edges of $v$ are computed.
\end{lemma}
% \begin{lemma}
% \label{lem:recover_computation}
% Let $T_i=(V_{T_i},E_{T_i})$ a tree (or fragment). Eventually for each node $v \in V_{T_i}$ the variable $\m_v$ contains a pair of values: the weight of an edge $(u,v)$ and the label of the nearest common ancestor of $u \in V_{T_i}$ and $v$, if the edge $(u,v) \not \in E_{T_i}$ is not part of a minimum spanning tree of the network. Moreover, eventually all local internal edges of $v$ are computed.
% \end{lemma}

% Developper preuve sur calcul de toutes les aretes internes
\begin{proof}
We assume that $T_i=(V_{T_i},E_{T_i})$ is a coherent tree, otherwise Rule $\RRoot$ and $\RLC$ are used to break the cycles and to correct the labels. Each node $v \in V_{T_i}$ starts to compute local internal edges (i.e., edges $(x,y)$ such that $x,y \in V_{T_i}$ and $x=v$ or $x$ or $y$ is in the subtree of $v$) when it has no local outgoing edge (adjacent outgoing edge or outgoing edge given by a child) leading to another tree. In this case, a node $v$ informs its parent of its local internal edges using its variable $\m_v$. Macro $\Update(v)$ returns the local internal edge of $v$ which has the common ancestor nearest from the root among the internal edges that were not taken into account by its parent using the node labels. Each node $v \in V_{T_i}$ which is in a recover phase sends all its local internal edges. To this end, $v$ compute its next local internal edge when its parent has taken into account $v$'s current local internal edge (i.e., $\m_{\parent_v}=\m_v$). Thus, the information of internal edges are put back up in the tree until reaching the nearest common ancestor and $v$ do not wait an acknowledgement from the nearest common ancestor to send its next internal edge. Moreover, Macro $\FarLca(v)$ compute the next internal edge adjacent to $v$ such that the label of the nearest common ancestor associated to $(u,v)$ with $u \in N(v)$ is greater (according to operator $>_{\lab}$) than $\m_v[1]$ which is used to define an order on the internal edges. Otherwise, if $\FarLca(v)=\emptyset$ then according to Macro $\UA(v)$ the node $v$ reset the computation of its adjacent internal edge to assure that every internal edge is taken into account. The recover phase is started at node $v$ if $v$ has no local outgoing edge (i.e., $\Fusion(v)=\emptyset$) and $v$'s parent has taken into account the information associated to its current internal edge and stored in variable $\m_v$ (i.e., $\m_{\parent_v}=\m_v$). In this case, the guard of Rule $\RMin$ is satisfied and $v$ can execute Rule $\RMin$ to update its variable $\m_v$ with the information of its next local internal edge. The information given by a child are stopped at the nearest common ancestor $v$ of the corresponding internal edge since Macro $\FarLcaC(v)$ selects only $\m_u$ from a child $u$ such that $\m_u[1] \leq_{\lab} \lab_v$.
\end{proof}

\begin{lemma}
\label{lem:mst_correction}
Let $T=(V_T,E_T)$ a tree and any edge $(x,y) \in E_T$. Eventually, if $(x,y)$ is not part of a minimum spanning tree of the network then $(x,y)$ is removed from $T$.
\end{lemma}

\begin{proof}
We assume that $T=(V_T,E_T)$ is a coherent tree, otherwise Rule $\RRoot$ and $\RLC$ are used to break the cycles and to correct node labels. Let an edge $(x,y) \in E_T$ (w.l.o.g. $\parent_y=x$) which is not part of a minimum spanning tree. As the network has a finite size then there is a time after which there exists no merging between two trees in the network. Thus according to Lemma~\ref{lem:recover_computation} each internal edge $e$ of $T$ is put back up in $T$ until reaching the nearest common ancestor associated to $e$. Since $(x,y)$ is not in a minimum spanning tree, there is an edge $(u,v)$ such that $w(u,v)<w(x,y)$ and $(x,y)$ is on the path between $u$ and $\Lca(u,v)$ or between $v$ and $\Lca(u,v)$. So, there is a time such that $\m_y=(w(u,v),\Lca(u,v))$ and $\m_x=(w(u,v),\Lca(u,v))$ according to Lemma~\ref{lem:recover_computation}. Then Predicate $\NewFragment(y)$ returns true because we have $\m_y=\m_x, \m_y[1]\neq \id_y$ and $w(x,y)>w(u,v)$. Therefore, $y$ can execute Rule $\RC$ to create a new tree rooted at $y$ and as a consequence the edge $(x,y)$ is removed from $T$.
\end{proof}

\begin{lemma}
\label{lem:merging_start}
Let two distinct trees $T_i=(V_{T_i},E_{T_i})$ and $T_j=(V_{T_j},E_{T_j})$ with $i\neq  j$. Let an edge $(x,y)$ such that $(x,y)$ is part of a minimum spanning tree and $x \in T_j$ and $y \in T_i$. Eventually $(x,y)$ is used to merge the trees $T_i$ and $T_j$.
\end{lemma}

\begin{proof}
We assume that $T_i=(V_{T_i},E_{T_i})$ and $T_j=(V_{T_j},E_{T_j})$ are coherent trees, otherwise Rule $\RRoot$ and $\RLC$ are used to break the cycles and to correct node labels. According to Lemma~\ref{lem:mwoe_computation}, the merging edge $(x,y)$ is computed by the root and it starts the merging phase since only the root can choose the edge of minimum weight leading to another tree to use in order to make a merging. In the remainder, we focus on tree $T_i$ but the same arguments are also true for $T_j$.

When the root $v$ has finished to compute its minimum outgoing edge from its fragment (i.e., we have $\m_v=\Fusion(v)$) then $v$ can execute Rule $\RF$ because Predicate $\NeedReorientation(v)$ is satisfied since $v$ has a coherent label. Note that we permit the creation of cycles of length two only if at least one node has a label equal to $(\bot,\bot)$. Indeed, during the merging phase the orientation is reversed on the path between the root of $T_i$ and the node adjacent to the edge used for the merging, that is why a cycle is detected in Rule $\RRoot$ if Predicate $\NeedReorientation(v)$ is not satisfied. Thus, $v$ can change its variables $\parent_v$ and $\lab_v$ as following. If there is an edge $(x,y)$ adjacent to $v$ such that $w(x,y)=\m_v[0]$ and $y=v$ then $v$ selects $x$ as its new parent (only if $\id_v>\id_x$) and $v$ changes its label to $(\emptyset,\emptyset)$ to informs its subtree that the merging is done. Otherwise, $v$ selects its child $u$ such that $\m_u=\m_v$ as its new parent and $v$ changes its label to $(\bot,\bot)$ to inform $u$ that a merging is started.\\
Any other node $v$ on the path between the root and the node adjacent to the merging edge take part in the merging phase when its parent has selected $v$ as its new parent (i.e., $\parent_{\parent_v}=\id_v$) and $v$'s parent label is equal to $(\bot,\bot)$. Thus, Predicate $\NeedReorientation(v)$ is satisfied and $v$ can execute Rule $\RF$ since $\m_v=\Fusion(v)$ (otherwise Rule $\RMin$ is executed to update its variable $\m_v$). So, $v$ changes its variables $\parent_v$ and $\lab_v$ as described above for the root. Since the merging phase is done on a path, using the same argument one can show by induction that for any internal node $v \in V_{T_i}$ on the path between the root and the node $y$ adjacent to the merging edge (except for $y$) we have $\parent_v=\min\{\id_u: u \in Child(v) \wedge \m_v=\Fusion(v)\}$ and $\lab_v=(\bot,\bot)$ and for the node $y$ we have $\parent_y=x$ and $\lab_y=(\emptyset,\emptyset)$.
\end{proof}

\begin{lemma}
\label{lem:merging_end}
Eventually all the nodes have a correct label in the new fragment resulting from a merging phase.
\end{lemma}

\begin{proof}
According to Lemma~\ref{lem:merging_start}, a merging phase is done using the minimum outgoing edge $(x,y)$ between two distinct trees $T_i$ and $T_j$ if it is possible. Moreover, when the edge $(x,y)$ is added by the extremity of minimum identifier, w.l.o.g. let $y \in$, then $y$'s label is equal to $(\emptyset,\emptyset)$ and the end of the merging phase is propagated in the resulting fragment $T$. In the reminder we focus on tree $T_i$ but the same arguments are true for tree $T_j$.

Let the node $v \in T_i$ such that $\parent_v=y$ and $\lab_v=(\bot,\bot)$ (i.e., $v$ is the child of $y$ on the path between $y$ and the old root of $T_i$). Predicate $\NeedReorientation(v)$ is false because $v$ is not a root node and the label of its parent $y$ is not equal to $(\bot,\bot)$. Moreover, Predicate $\EndReorientation(v)$ is true since $y$'s label is equal to $(\emptyset,\emptyset)$ and $v$'s label to $(\bot,\bot)$. Thus, $v$ can execute Rule $\REnd$ to modify its label to $(\emptyset,\emptyset)$. Using the same argument, one can show by induction that every node $v$ on the path between $y$ and the old root of $T_i$ can execute Rule $\REnd$, thus there is a time such that we have $\lab_v=(\emptyset,\emptyset)$.\\
Now we show that the other nodes in $T_i$ can execute Rule $\REnd$. Consider the node $r \in V_{T_i}$ such that $r$ is the old root of $T_i$ and $\lab_z=(\emptyset,\emptyset)$. Let a node $v \in V_{T_i}$ such that $\parent_v=\id_z$. Predicate $\NeedReorientation(v)$ is false because $v$ is not a root node and the label of its parent is not equal to $(\bot,\bot)$. Moreover, Predicate $\EndReorientation(v)$ is true because $\lab_z=(\emptyset,\emptyset)$ and $v$'s label is not equal to $(\bot,\bot)$ since $v$ is not on the path between $y$ and the old root of $T_i$, and $v$'s label is different from $(\emptyset,\emptyset)$. Note that since the start of the merging phase, Predicate $\TreeChange(v)$ is true because Predicate $\NeedReorientation(v)$ or $\EndReorientation(v)$ is true. So, Rule $\RLC$ cannot be executed by $v$ and $v$'s label has not changed. Thus, $v$ can execute Rule $\REnd$ to modify its label to $(\emptyset,\emptyset)$. Using the same argument, one can show by induction that every node $v$ on the path between $z$ and a leaf node can execute Rule $\REnd$, thus there is a time such that we have $\lab_v=(\emptyset,\emptyset)$.

Every node $v$ in the resulting fragment $T$ can execute Rule $\RLC$ when the edge $(x,y)$ is added in $T$ and $y$'s label has been modified from $(\emptyset,\emptyset)$ to its new label based on $x$'s label. Indeed, in this case for every node $v$, with $v \neq y$ and $v \in V_{T_i}$, Predicate $\TreeChange(v)$ is false and $v$ can execute Rule $\RLC$. Therefore, there is a time such that every node $v$ in $T$ has a correct label.
\end{proof}

\begin{lemma}[Convergence for $\MST$]
\label{lem:mst_convergence}
Starting from an illegitimate configuration for Algorithm $\MST$, eventually Algorithm \MST\/ reaches in a finite time a legitimate configuration.
\end{lemma}

\begin{proof}
We assume that there is a forest of trees $T_i, 1 \leq i \leq n$, in the network, otherwise according to Lemma~\ref{lem:no_cycle} Rule $\RRoot$ is executed to remove the cycle from the network. Moreover, we assume also that the node's label are correct in tree $T_i$, otherwise according to Lemma~\ref{lem:label_convergence} and~\ref{lem:label_closure} there is a time such that the node's label are corrected.

According to Lemma~\ref{lem:recover_computation} and~\ref{lem:mst_correction}, if an edge $(u,v) \in E_{T_i}$ and $(u,v)$ is part of no minimum spanning tree of the network then $(u,v)$ is removed from tree $T_i$. Thus, there is a time such that the existing fragments in the network are part of a minimum spanning tree. According to Lemmas~\ref{lem:mwoe_computation} and~\ref{lem:merging_start}, eventually if there are at least two distinct fragments then a merging phase is started. Moreover, node labels are corrected after a merging phase according to Lemma~\ref{lem:merging_end}. Since the size of the network is finite there is a finite number of merging. Therefore, in a finite time a spanning tree of minimum weight is computed by Algorithm $\MST$.
\end{proof}

\begin{lemma}[Closure for $\MST$]
\label{lem:mst_closure}
The set of legitimate configurations for \MST\/ is closed.
\end{lemma}

\begin{proof}
Let $\mathcal{C}$ a legitimate configuration such $T=(V_T,E_T)$ is a minimum spanning tree of the network and an edge $(x,y) \in E_T$. To be illegitimate, the configuration $\mathcal{C}$ must contain an edge $(x,y)$ such that it exists an edge $e=(u,v)$ with $w(u,v)<w(x,y)$ and $(x,y)$ and $(u,v)$ are included in the same fundamental cycle $C_e$. Thus, this imply that the edge $e$ is not used to verify if it is possible to replace an edge of $C_e$ with $(u,v)$ which contradicts Lemmas~\ref{lem:recover_computation} and~\ref{lem:mst_correction}. Moreover, since $T$ is a spanning tree then no merging is done in the network. Therefore, starting from a legitimate configuration for Algorithm \MST\/ a legitimate configuration is preserved.
\end{proof}

\section{Complexity proofs}
In the following we discuss the complexity issues of our solution.

\begin{lemma}
Algorithms \LabA\/ and \MST\/ have a space complexity of $O(\log^2n)$ bits.
\end{lemma}

\begin{proof}
Algorithm \LabA\/ uses three variables : $\parent_v, \lab_v,
\TS_v$. The first and the last one are respectively a pointer to 
a neighbor node and a pair of integers, each one needs $O(\log n)$
bits. However, the variable $\lab_v$ is a list 
of pairs of integers. A new pair of integers is added to the list when
a light edge is created in the tree. 
As noticed in~\cite{AGKR02}, there are at most $\log n$ light edges on
the path from a leaf to the root, i.e., 
at most $\log n$ pairs of integers. Thus, the variable $\lab_v$ uses $\log n \times \log n$ bits.

Algorithm \MST\/ uses an additional variable $\m_v$ which is a pair
composed of an integer and the label of a node. 
The label of a node is stored in variable $\lab_v$ which uses $\log^2 n$ bits. Thus, the variable $\m_v$ needs $\log^2 n$ bits.

Therefore, Algorithms \LabA\/ and \MST\/ use $O(\log^2n)$ bits of memory at each node.
\end{proof}

% \begin{lemma}
% \label{lem:no_cycle_complexity}
% Starting from any configuration, the cycle $C_k$ is removed from the network in at most $O(n_k)$ rounds, with $n_k$ the number of nodes in $C_k$.
% \end{lemma}
\begin{lemma}
\label{lem:no_cycle_complexity}
Starting from any configuration, all cycles are removed from the network in at most $O(n^2)$ rounds, with $n$ the number of nodes in the network.
\end{lemma}

\begin{proof}
As explained in the proof of Lemma~\ref{lem:no_cycle}, to break a cycle $C_k$ a part of the nodes in $C_k$ must compute their new labels, that is a label computation must be initiated from one node and then this process must cross $C_k$. Thus, the worst case is a configuration in which all the nodes in $C_k$ have to compute their new labels using Rule $\RLC$ to detect the presence of cycle $C_k$. Therefore, at most $O(n)$ rounds are needed to compute the new label of the nodes in $C_k$ based on the label of one node $x$ in $C_k$. According to Lemma~\ref{lem:no_cycle}, when this computation is done the cycle $C_k$ is detected and removed by the node $x$. At most $O(n)$ additional rounds are needed to break the cycle $C_k$.

Since there is at most $n/2$ cycles in a network, at most $O(n^2)$ rounds are needed to remove all the cycles from the network.
\end{proof}

\begin{lemma}
\label{lem:label_complexity}
Starting from a configuration which contains a tree $T$, using Algorithm \LabA\/ any node $v \in V_T$ has a correct label in at most $O(n)$ rounds.
\end{lemma}

\begin{proof}
As described in proof of Lemma~\ref{lem:label_convergence}, the correction of node labels is done using a bottom-up computation followed by a top-down computation in the tree $T=(V_T,E_T)$.

The bottom-up computation is started by the leaves of $T$, when leaf nodes $v \in V_T$ have corrected their variable $\TS_v$ to $(1,\bot)$ then internal nodes $u \in V_T$ can start to correct their variable $\TS_u$. An internal node $v \in V_T$ computes a correct value in its variable $\TS_v$ using Rule $\RLC$ when all its children $u$ have a correct value in their variable $\TS_u$. Since the computation is done in a tree sub-graph then in at most $O(n)$ rounds each node $v \in V_T$ has corrected its variable $\TS_v$.

The top-down computation is started by the root of the tree $T$. When the root $v$ has a correct value in variable $\TS_v$ then the computation of correct labels can start. Thus, if the parent of a node $v$ has a correct value in its variable $\TS_{\parent_v}$ and $\lab_{\parent_v}$ then $v$ can compute its correct label in $\lab_v$ using Rule $\RLC$. As for the bottom-up computation, the top-down computation is done in at most $O(n)$ rounds since it is performed in a tree sub-graph.

Therefore, in at most $O(n)$ rounds each node $v$ in the tree $T$ has a correct label stored in variable $\TS_v$.
\end{proof}

\begin{lemma}
\label{lem:label_total_complexity}
Starting from any configuration, Algorithm \LabA\/ reaches a legitimate configuration in at most $O(n^2)$ rounds.
\end{lemma}

\begin{proof}
The initial configuration $\mathcal{C}$ could contain one or more cycles, so according to Lemma~\ref{lem:no_cycle_complexity} in at most $O(n^2)$ rounds the system reaches a new configuration $\mathcal{C'}$ which contains no cycle. Moreover according to Lemma~\ref{lem:label_complexity}, the nodes $v$ in each tree $T$ in the configuration $\mathcal{C'}$ have a correct label in at most $O(n)$ rounds. Therefore, starting from an arbitrary configuration each node $v \in V$ computes its correct label in at most $O(n^2)$ rounds.
\end{proof}

\begin{lemma}
\label{lem:mst_complexity}
Starting from any configuration, Algorithm \MST\/ reaches a legitimate configuration in at most $O(n^2)$ rounds.
\end{lemma}

\begin{proof}
% -> Suppression arête (O(n^2) rounds)
% -> Phase de fusion
% 	-labeling des noeuds dans arbre (O(n) rounds)
% 	-remonte arête mwoe (O(n) rounds)
% 	-début fusion (O(n) rounds)
% 	-fin fusion (O(n) rounds)
% -> envoie arête de correction dans arbre (O(n^2) rounds)
% 	-suppression arête (temps constant)
According to Lemma~\ref{lem:no_cycle_complexity}, starting from any configuration after at most $O(n^2)$ rounds all the cycles are removed from the network, i.e., it remains a forest of trees after at most $O(n^2)$ rounds. Moreover, according to Lemma~\ref{lem:label_complexity} in at most $O(n)$ additional rounds each node $v \in V$ has a correct label since each node belongs to a unique tree.

% As described in Lemma~\ref{lem:mwoe_computation}, when it is possible to make a merging between two distinct trees in the forest a merging phase is started. This merging phase is done in three steps: (1) information corresponding to the minimum outgoing edge is propagated in a bottom-up fashion in each tree, (2) the orientation is reversed from the root of a tree until reaching the node in the tree adjacent to the minimum outgoing edge, and (3) the node labels are changed to inform of the end of the merging phase, followed by a propagation of the new correct node labels in the new tree resulting from the merging phase.

According to the description of Algorithm $\MST$, Macro $\MinEdge(v)$ and Lemma~\ref{lem:recover_computation}, when it is possible to make a merging between two distinct trees in the forest a merging phase is started. This merging phase is done in three steps: (1) information corresponding to the minimum outgoing edge is propagated in a bottom-up fashion in each tree, (2) the orientation is reversed from the root of a tree until reaching the node in the tree adjacent to the minimum outgoing edge, and (3) the node labels are changed to inform of the end of the merging phase, followed by a propagation of the new correct node labels in the new tree resulting from the merging phase.

The first step is a propagation of information in a bottom-up fashion in a tree which is done in at most $O(n)$ rounds. The second step reverses and propagates new node labels on a part of the tree (between the root and the node adjacent to the minimum outgoing edge) which is done in at most $O(n)$ rounds too. Step 3 modifies the label of the nodes which have changed their parent pointer in step 2, so this last step takes also at most $O(n)$ rounds and the relabeling of the nodes in the new tree is done in at most $O(n)$ rounds according to Lemma~\ref{lem:label_complexity}. Thus, a merging phase is accomplished in at most $O(n)$ rounds and as there are in the worst case $n$ trees then in at most $O(n^2)$ rounds a spanning tree is constructed.

When there is no possible merging for a given fragment (or tree) $T_i$ in the forest then the correction phase concerning $T_i$ is started. In a tree $T_i$, the internal edges (i.e., whose two endpoints are in $T_i$) are sent upward in $T_i$ in order to detect incorrect tree edges. The internal edges $e$ are sent following an order on the distance between the common ancestor $\Lca(e)$ and the root of $T_i$, by sending first the edge $e$ with the nearest common ancestor $\Lca(e)$ from the root. Let $h(T_i)$ be the height of tree $T_i$ and $d(v)$ be the distance from $v \in T_i$ to the root of $T_i$. Thus, an internal (resp. leaf) node has at most $d(v)-2$ (resp. $d(v)-1$) adjacent internal edges. Since a leaf node could have a lower priority (compared to its ancestors) to send all its adjacent internal edges, then the worst case to correct a tree is the case of a chain. Indeed, if the last internal edge of a leaf node $x$ must be used to detect an incorrect tree edge then $x$ may have to wait that all its ancestors in the chain have sent their internal edges of higher priority. Thus, starting from any configuration after at most $O(h(T_i)^2)$ rounds $T_i$ contains no incorrect edges. Note that this is the worst case time to detect the farthest incorrect tree edge from the root of $T_i$, otherwise the correction phase is stopped earlier for nearest incorrect tree edges because the merging phase has a higher priority than the correction phase. Moreover, after $O(h(T_i)^2)$ rounds all the new edges used by $T_i$ for a merging are correct tree edges for $T_i$. So, $T_i$ does not remove another tree edge in a new correction phase. Hence starting from any configuration, a correction phase deletes all the incorrect tree edges of a spanning tree after at most $O(n^2)$ rounds and no new tree edges are removed by a correction phase.

Therefore, starting from an arbitrary configuration Algorithm $\MST$ constructs a minimum spanning tree in at most $O(n^2)$ rounds.
\end{proof}

\section{Conclusion}
We extended the Gallager, Humblet and Spira (GHS) algorithm, 
\cite{GallagerHS83}, 
to self-stabilizing settings via a compact informative labeling scheme.
Thus, the resulting solution presents several advantages 
appealing for large scale systems: it is compact since it 
uses only logarithmic memory in the size of the network, 
it scales well since it does not rely on any global 
parameter of the system, it is fast --- its time complexity is the 
better known in self-stabilizing settings. Additionally, it self-recovers 
from any transient fault. The time complexity is $O(n^2)$
rounds and the space complexity is $O(log^2n)$.

%====================
% \newpage

\end{document}